\documentclass[a4paper, amsfonts, amssymb, amsmath, reprint, showkeys, nofootinbib, twoside,notitlepage,onecolumn]{revtex4-1}

\usepackage{amsmath,amstext}
\usepackage[T1]{fontenc}
\usepackage{amssymb}
\usepackage{graphicx}
\usepackage{ae,aecompl}

\DeclareFontFamily{OT1}{pzc}{}
\DeclareFontShape{OT1}{pzc}{m}{it}{<-> s * [1.10] pzcmi7t}{}
\DeclareMathAlphabet{\mathpzc}{OT1}{pzc}{m}{it}

\usepackage{hyperref}
\usepackage{amsmath}
\usepackage{amssymb}
\usepackage{mathtools}
\usepackage{bm}
\usepackage{cleveref}
\usepackage{tensor}
\usepackage{braket}
\usepackage{enumitem}
\usepackage{mhchem}
\usepackage{amsthm}
\usepackage{nccmath}
\usepackage{mathrsfs}
\usepackage{color}

\newcommand{\spc}{\quad \quad \quad}

\newcommand{\E}{{\mathbb{E}}}
\newcommand{\G}{{\mathbb{G}}}
\newcommand{\PC}{{\mathbb{P}_C}}
\newcommand{\PN}{{\mathbb{P}_N}}
\newcommand{\s}{{\sigma \!\! \! \sigma}}

\def\be{\begin{equation}}
\def\ee{\end{equation}}
\def\beq{\begin{eqnarray}}
\def\eeq{\end{eqnarray}}

\theoremstyle{definition}

\theoremstyle{theorem}
\newtheorem{theorem}{Theorem}

\theoremstyle{corollary}

\begin{document}
\title{The causal diamond of relativistic hydrodynamics}
\author{L.~Gavassino}
\affiliation{Department of Applied Mathematics and Theoretical Physics, University of Cambridge, Wilberforce Road, Cambridge CB3 0WA, United Kingdom}

\begin{abstract}
The convergence of hydrodynamics is commonly investigated by Taylor-expanding the gapless dispersion relations $\omega(k)$ around $k=0$. We show that this diagnostic can fail arbitrarily close to the hydrodynamic limit: sound modes can collide at arbitrarily small wave number, rendering $\omega(k)$ nonanalytic, while the relationship between densities and fluxes remains analytic there. We instead formulate the gradient expansion off-shell (i.e. in the presence of external forces), directly at the level of the constitutive relations. For causal theories with relaxational spectra, we establish convergence for $|i\omega|+w|ik|<1/\tau_g$, where $w$ is the maximal information speed and $1/\tau_g$ the nonhydrodynamic gap. On the naturally Lorentzian $\{i\omega,ik\}$ plane, this region is a causal diamond, whose boundary marks the first possible encounter with nonhydrodynamic modes. Beyond it, analytic continuation reveals the nonhydrodynamic spectrum through singularities of the resummed constitutive relations. Thus, hydrodynamics is aware of the degrees of freedom beyond hydrodynamics, but this information becomes fully accessible only off-shell.
\end{abstract} 
\maketitle

\noindent \textbf{\textit{Introduction ---}} The fact that hydrodynamics can be organized as a gradient expansion is by now common knowledge \cite{Baier2008,RomatschkeReview:2017ejr,FlorkowskiReview2018,Kovtun2019,BemficaDNDefinitivo2020,GavassinoInfiniteOrded2024pgl}. Within linear-response theory, this means that, if $\rho_I$ are perturbations of conserved densities and $J_I$ the associated fluxes, then near the origin in Fourier space one can write a formal Taylor series
\begin{equation}\label{JIgeneral}
J_I =\sum_{H,n,m} \chi_{IHnm} (i\omega)^n (ik)^m \rho_H \spc
(\text{for }|\omega|,|k|\text{ sufficiently small}) .
\end{equation}
These constitutive relations admit a more familiar representation. Separating the zeroth-order contribution, and noting that diffusive transport vanishes in the absence of spatial gradients, we can rewrite \eqref{JIgeneral} as
\begin{equation}\label{decomposco}
J_I =
\underbrace{\sum_H\chi_{IH}\rho_H}_{\text{Local equilibrium}}
-\underbrace{ik\sum_H\mathfrak{D}_{IH}(i\omega,ik)\rho_H}_{\text{Diffusive part}} \, .
\end{equation}
The functions $\mathfrak{D}_{IH}(i\omega,ik)$ are holomorphic near the origin and may be interpreted as \textit{resummed diffusivities}. 

Such frequency- and wave-number-dependent transport coefficients are familiar from generalized hydrodynamics, rheology, and nonlocal electronic transport \cite{ChungYip1969,Hess1983,TGaskell_1987,VelascoColin,Hansen2007,AlfordBulk,Bu:2014sia,StephanovHydroPlus:2017ghc,Kiselev,BulkGavassino,HellerSingulant:2021yjh,Krihin2023,Nazaryan2024}.
Yet, little is known about how relativity, and causality in particular, constrains the analytic structure of $\mathfrak{D}_{IH}(i\omega,ik)$. Instead, much of the relativistic linear-response literature has focused on Taylor expansions of individual hydrodynamic (i.e., gapless) dispersion relations, $\omega(k)=\sum_n a_n k^n$ \cite{Grozdanov:2019kge,GrozdanovComplexlife:2019uhi,HellerLinearresponse:2020uuy,HellerConvergencemodes:2020hnq,Liu:2021qmt,HellerBounds2022ejw,HellerHydrohedron2023jtd}. These provide a useful proxy for the gradient expansion, since the dispersion relations follow from solving $-i\omega\rho_I+ikJ_I=0$, but have two important limitations. First, they contain no information about $J_I$ at values of $\omega$ and $k$ that do not satisfy the source-free conservation laws but remain accessible through external forcing (i.e. in the ``off-shell regime'' \cite{Haehl:2015pja,Bu:2015ame,Bu:2018drd,Diles:2023tau}). Second, and more importantly, we will show that the Taylor series of $\omega(k)$ can break down arbitrarily close to the origin even when $\mathfrak{D}_{IH}(i\omega,ik)$ remains holomorphic over a much larger region of $(\omega,k)$-space.

Here we study the effect of relativity on $\mathfrak{D}_{IH}(i\omega,ik)$ within kinetic theory, transient hydrodynamics, and rheology, and more generally in the class of theories discussed in \cite{GavassinoDisturbing:2026klp,GavassinoLorentzian:2026seq}. We prove that causality guarantees convergence of the resummed diffusivities within a universal domain set by the maximal information speed and the nonhydrodynamic gap. On the relaxation plane, this domain is a causal diamond, corresponding to the region where non-hydrodynamic spectral information cannot enter. We further derive universal bounds on all the coefficients of the gradient expansion \eqref{JIgeneral}, and illustrate the results with explicit examples.

Throughout the article, we adopt the metric signature $(-,+,+,+)$, and work in natural units, $c=\hbar=k_B=1$.

\noindent \textbf{\textit{Gradient expansion vs dispersion relations ---}} We begin by illustrating why the convergence of the gradient expansion cannot, in general, be inferred from that of the hydrodynamic dispersion relations. Consider Israel-Stewart theory of bulk viscosity at zero chemical potential. In the sound channel, the relevant conserved densities are the energy and longitudinal-momentum densities, $\rho_I{=}(\delta\varepsilon,\delta\mathcal{F})$, with associated fluxes $J_I=(\delta\mathcal{F},c_s^2\delta\varepsilon+\delta\Pi)$. Here, $c_s$ is the speed of sound and $\delta\Pi$ is the bulk pressure, which obeys the Israel-Stewart relaxation equation $\tau_g\partial_t\delta\Pi+\delta\Pi=-\mathfrak{D}\partial_x\delta u$. Working in Fourier space, we immediately see that $\mathfrak{D}_{\varepsilon\varepsilon}=\mathfrak{D}_{\varepsilon\mathcal{F}}=\mathfrak{D}_{\mathcal{F}\varepsilon}=0$, while
\begin{equation}\label{constituoIS}
\mathfrak{D}_{\mathcal{F}\mathcal{F}}(i\omega,ik)=\frac{\mathfrak{D}}{1-i\omega\tau_g}
\equiv \mathfrak{D}\sum_n(i\omega\tau_g)^n
\spc (\text{for }|\omega\tau_g|<1) .
\end{equation}

\noindent Let us compare the convergence domain of this geometric series with that of the Taylor series of the hydrodynamic modes. Using the conservation laws $-i\omega\rho_I+ikJ_I=0$, and introducing the characteristic speed $w^2=c_s^2+\mathfrak{D}/\tau_g$, we obtain $(i\omega\tau_g)^3-(i\omega\tau_g)^2+(wk\tau_g)^2i\omega\tau_g-(c_sk\tau_g)^2=0$. Its roots consist of two sound modes and one nonhydrodynamic mode. For sufficiently small $k$, the discriminant $\Delta=-4c_s^2(k\tau_g)^2+(-27c_s^4{+}18c_s^2w^2{+}w^4)(k\tau_g)^4-4w^6(k\tau_g)^6$ is negative, corresponding to two propagating sound modes. However, for $c_s\,{\ll}\, w$, it changes sign at $k_c\tau_g\,{\simeq}\,2c_s/w^2$, where the two sound poles collide and become nonpropagating (see Fig.~\ref{fig:IS}). Hence, by taking $c_s\to0$, the collision can be brought arbitrarily close to the origin (see also \cite{AbbasiKaminski}).

Crucially, nothing singular happens to the constitutive relation \eqref{constituoIS} at $k=k_c$: $\mathfrak{D}_{\mathcal{F}\mathcal{F}}$ remains perfectly analytic there, with a domain of convergence $|\omega\tau_g|<1$ independent of $c_s$ and $k$. The nonanalyticity arises only upon solving the conservation laws for the individual branches $\omega(k)$, thereby introducing algebraic branch points that are absent from the constitutive relations themselves. Thus, the radius of convergence of $\omega(k)$ can become arbitrarily small without any corresponding breakdown of \eqref{JIgeneral}. In the Supplementary Material, we show that such sound-sound collisions at $k_c=\mathcal{O}(c_s)$ can persist in richer kinetic models, where $\mathfrak{D}_{\mathcal{F}\mathcal{F}}$ has infinitely many spatial-gradient corrections.

\begin{figure}[t]
    \centering
\includegraphics[width=0.38\linewidth]{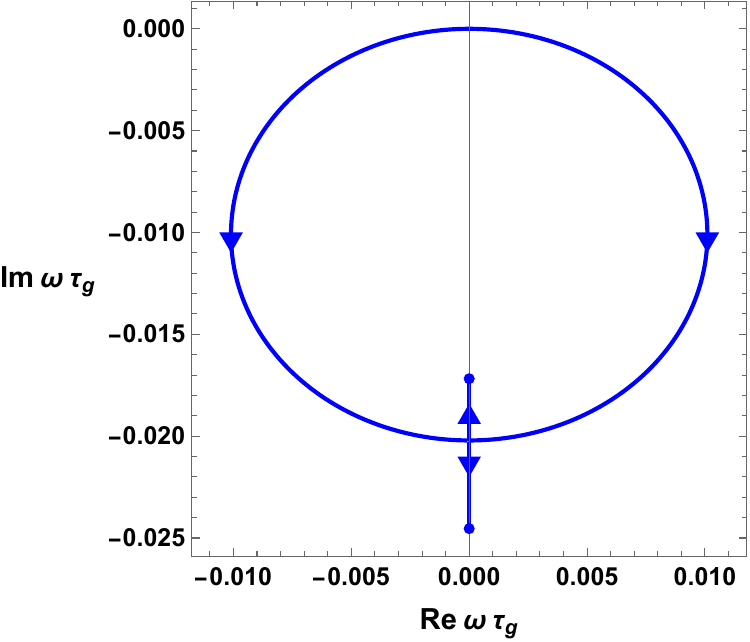}\hspace{0.08\linewidth}
\caption{Sound modes of Israel-Stewart theory, with $w=1$ and $c_s=0.1$. As $k>0$ increases, the two sound modes trace an approximately circular orbit and collide at $k\tau_g\simeq0.2$ and $|\omega\tau_g|\simeq0.02$. The collision thus occurs well within the domain $|\omega\tau_g|<1$, where the gradient expansion of the constitutive relation \eqref{constituoIS} converges.}
    \label{fig:IS}
\end{figure}

Having seen that the analytic structure of individual dispersion relations can obscure the true convergence properties of \eqref{JIgeneral}, we now turn directly to the constitutive relations and ask what can be learned about their gradient expansion from the underlying microscopic dynamics.

\noindent \textbf{\textit{Assumptions on the microscopic theory ---}} It has been shown \cite{GavassinoSymmetric2022nff,GavassinoConvergence:2024xwf,GavassinoDisturbing:2026klp} that kinetic theory, transient hydrodynamics, and causal viscoelasticity can all be recast in the universal form
\begin{equation}\label{boltzmann}
\partial_t\Psi=-(\s+\E\partial_x)\Psi+Q \, ,
\end{equation}
where $\Psi:\text{``Spacetime''}{\to}\mathcal{H}$ is a linearized degree of freedom taking values in a Hilbert space $\mathcal{H}$, and $Q:\text{``Spacetime''}{\to}\mathcal{H}$ is an external source. The inner product $(\cdot,\cdot)$ on $\mathcal{H}$ is Onsager-normalized, so that the quadratic free-energy perturbation is $\Delta\mathfrak{F}=\frac{1}{2}(\Psi,\Psi)$. Assuming that $\Psi$ is even under PT symmetry, Onsager reciprocity implies that $\s$ and $\E$ are self-adjoint. Furthermore, causality implies that $\E$ is bounded, with norm $w=||\E||\leq1$ defining the maximal speed of information. Finally, we assume that $0$ is an eigenvalue of $\s$ with finite multiplicity and that the nonhydrodynamic spectrum is gapped, $\text{Spectrum}(\s)\subseteq\{0\}\cup[1/\tau_g,+\infty)$, where $1/\tau_g$ is the nonhydrodynamic gap  \cite{Dudynski1989Hydro,GavassinoGapless:2024rck,GavassinoMemoryEffects:2026ebx}.

\noindent \textbf{\textit{Gradient expansion from microphysics ---}} We consider the orthogonal decomposition $\mathcal{H}=\mathcal{H}_C \oplus \mathcal{H}_N$, where $\mathcal{H}_C=\ker(\s)$ and $\mathcal{H}_N=[\ker(\s)]^\perp$, with associated orthogonal projectors $\PC$ and $\PN$. This decomposition separates the degrees of freedom into their ``conserved'' and ``nonconserved'' parts. In particular, let $e_I$ be an orthonormal basis of $\ker(\s)$. Taking the inner product of \eqref{boltzmann} with $e_I$, we immediately obtain the conservation laws $\partial_t \rho_I+\partial_x J_I=(e_I,Q)$, with $\rho_I=(e_I,\Psi)$ and $J_I=(e_I,\E\Psi)$. Thus, writing $\Psi=\PC\Psi+\PN\Psi=\sum_H\rho_H e_H+\Psi_N$, the term $\sum_H\rho_H e_H\in\mathcal{H}_C$ encodes the density content of the state, while $\Psi_N\in\mathcal{H}_N$ is its nonequilibrium part. Indeed, any constant state of the form $\Psi=\sum_I c_I e_I$ is always a solution of \eqref{boltzmann} with $Q=0$, corresponding to the global equilibrium state obtained by shifting the conserved densities by $c_I$.

Now, let us assume that $Q$ takes the form of an exponential plane wave, $Q\propto e^{ikx-i\omega t}$, where $\omega$ and $k$ are arbitrary and may be complex. We can then seek solutions of \eqref{boltzmann} with $\Psi\propto e^{ikx-i\omega t}$. Applying $\PN$ to both sides, we obtain
\begin{equation}\label{PsiN}
(\s_N+ik\E_N-i\omega)\Psi_N=-ik\sum_H\PN\E e_H\rho_H +\PN Q \, ,
\end{equation}
where we have introduced the compressions $\s_N=\PN\s\PN$ and $\E_N=\PN\E\PN$, both viewed as operators on $\mathcal{H}_N$. Suppose that $i\omega$ belongs to the resolvent set of $\s_N+ik\E_N$. Then, the operator $\s_N+ik\E_N-i\omega$ is invertible with bounded inverse, allowing us to solve \eqref{PsiN} for $\Psi_N$. Substituting the result into $J_I=(e_I,\E\Psi)$, and reabsorbing the term $\propto \PN Q$ into the source term on the right-hand side of the conservation law, we recover \eqref{decomposco}, with
\begin{equation}\label{chianddtheboss}
\chi_{IH}=(e_I,\E e_H) \, , \qquad \qquad \mathfrak{D}_{IH}(i\omega,ik)= (\PN\E e_I,[\s_N+ik\E_N-i\omega]^{-1}\PN\E e_H) \, .
\end{equation}
We can now state our main result.

\begin{theorem}\label{theo1}
Let $\mathfrak{D}_{IH}(i\omega,ik)$ be defined as above. These functions are holomorphic for complex $\omega$ and $k$ throughout the region $|i\omega|+|wik|<1/\tau_g$, where they admit an absolutely convergent expansion $\mathfrak{D}_{IH}=\sum_{n,m}\mathfrak{D}_{IHnm}(i\omega)^n(ik)^m$. Furthermore, within this region, the following bounds hold:
\vspace{-0.1cm}
\begin{equation}\label{bounds}
\begin{split}
|\mathfrak{D}_{IH}(i\omega,ik)| &\leq \dfrac{w^2\tau_g}{1-|i\omega\tau_g|-|wik\tau_g|} \qquad \qquad (\text{for }\,\,|i\omega|+|wik|<1/\tau_g) \, , \\
|\mathfrak{D}_{IHnm}| &\leq \binom{n{+}m}{m}w^{m+2}\tau_g^{n+m+1} \, .
\end{split}
\end{equation}
\end{theorem}

\begin{proof}
Since $\text{Spectrum}(\s_N)\subseteq[1/\tau_g,\infty)$, the inverse $\s_N^{-1}$ exists, and $||\s_N^{-1}||\leq\tau_g$. Thus, we can write $[\s_N{+}ik\E_N{-}i\omega]^{-1}=[1-\s_N^{-1}(i\omega{-}ik\E_N)]^{-1}\s_N^{-1}$. Using the properties of the operator norm, we find $||\s_N^{-1}(i\omega{-}ik\E_N)||\leq\tau_g(|i\omega|{+}w|ik|)$, so the Neumann series of $[1-\s_N^{-1}(i\omega{-}ik\E_N)]^{-1}$ converges absolutely for $|i\omega|{+}|wik|<1/\tau_g$, and we can write
\vspace{-0.1cm}
\begin{equation}\label{DIHiwik}
\mathfrak{D}_{IH}(i\omega,ik)=\sum_{a=0}^\infty(\PN\E e_I,[\s_N^{-1}(i\omega-ik\E_N)]^a\s_N^{-1}\PN\E e_H) \, .
\end{equation}
From this, we immediately obtain $|\mathfrak{D}_{IH}|\leq\tau_gw^2\sum_a\tau_g^a(|i\omega|+|wik|)^a$, which, once resummed, yields the first line of \eqref{bounds}. Expanding the powers of $[\s_N^{-1}(i\omega-ik\E_N)]^a$ and collecting together all terms proportional to $(i\omega)^n(ik)^m$, we obtain the series $\sum_{n,m}\mathfrak{D}_{IHnm}(i\omega)^n(ik)^m$. Each $\mathfrak{D}_{IHnm}$ is a sum of matrix elements $(\PN\E e_I,\text{``String of operators''}\,\s_N^{-1}\PN\E e_H)$, where the string of operators contains $n+m$ factors of $\s_N^{-1}$ and $m$ factors of $\E_N$. Each matrix element is therefore bounded in magnitude by $w^{m+2}\tau_g^{n+m+1}$, while the binomial coefficient in \eqref{bounds} counts the number of such terms.
\end{proof}

\noindent \textbf{\textit{The causal diamond ---}} The inequality $|i\omega|+|wik|<1/\tau_g$ is difficult to interpret if we restrict $\omega$ and $k$ to be real. Indeed, it can be shown (see Supplementary Material) that $\mathfrak{D}_{IH}(i\omega,ik)$ remains analytic everywhere on the Fourier plane. The physics becomes transparent instead on the relaxation plane, where $i\omega$ and $ik$ are real. As shown in \cite{GavassinoLorentzian:2026seq}, this plane inherits a geometry mirroring that of spacetime: just as signals propagate along worldlines timelike with respect to the acoustic metric $ds^2=-w^2dt^2+dx^2$, dispersion relations trace lines spacelike with respect to the dual acoustic metric $d\Tilde{s}^2=-d(i\omega)^2+w^2d(ik)^2$. In this geometry, $|i\omega|+|wik|<1/\tau_g$ is the causal diamond generated by the vertical segment $(i\omega,ik)\in[-1/\tau_g,1/\tau_g]\times\{0\}$. Crucially, singularities of $\mathfrak{D}_{IH}$ do appear once we leave this diamond. Indeed, for real $ik$, the compressed operator $\s_N+ik\E_N$ is self-adjoint, so its spectrum is real, and the resummed diffusivities generically become singular precisely when $i\omega\in \text{Spectrum}(\s_N+ik\E_N)$. At $ik=0$, this spectrum starts above the generating segment and, by causality, cannot propagate inside the diamond. It can, however, graze its upper boundary. In Israel-Stewart theory this happens only at the tip, $i\omega=1/\tau_g$, whereas in kinetic theory ballistic excitations can trace the upper horizon, as we will see explicitly in the examples. This also explains the divergence of the bound in \eqref{bounds}: on the relaxation plane, $||[\s_N+ik\E_N-i\omega]^{-1}||$ is the inverse distance of $i\omega$ from the spectrum, and therefore diverges on the upper horizon whenever the nonhydrodynamic spectrum grazes it.

\noindent \textbf{\textit{Hydrodynamics knows about the nonhydrodynamic modes ---}} For Bjorken flow, it has been shown that the all-order gradient expansion encodes the nonhydrodynamic spectrum
\cite{Heller:2013fn,HellerSplainskiResurgence2015,AnicetoSpalinski,HellerKurkela,Aniceto:2018uik}.
Here, Theorem~\ref{theo1} provides the linear-response counterpart of this result.

Suppose that we know all coefficients $\mathfrak{D}_{IHnm}$. Then, we can reconstruct the resummed diffusivity $\mathfrak{D}_{IH}(i\omega,ik)$ within the region $|i\omega|+|wik|<1/\tau_g$. At fixed $k$, we can then analytically continue $\mathfrak{D}_{IH}(i\omega,ik)$ in $\omega$ beyond the convergence domain of the series, until a singularity is encountered. We know that such singularities can only occur for $i\omega\in\text{Spectrum}(\s_N{+}ik\E_N)$. At $k=0$, we have $\text{Spectrum}(\s_N)=\text{Spectrum}(\s)\setminus\{0\}$, so these singularities lie precisely on the nonhydrodynamic spectrum (more precisely, on the part of it that couples to the constitutive relations).

For $k\neq0$, the situation is more subtle, because $\text{Spectrum}(\s_N+ik\E_N)$ need not be a subset of $\text{Spectrum}(\s+ik\E)$. We already encountered this phenomenon in the Israel-Stewart example (equation \eqref{constituoIS}): the resummed diffusivity diverges at $\omega=-i/\tau_g$ for every $k$, even though at finite $k$ the true nonhydrodynamic mode lies elsewhere. However, the compressed and full operators have the same \textit{essential} spectrum. Indeed, since $\PC$ has finite rank, the difference induced by the compression amounts to the finite-rank operator $\mathbb{K}\propto\PC\E\PC-\PC\E-\E\PC$, which cannot modify the essential spectrum. Thus, although isolated nonhydrodynamic modes need not be reconstructed correctly at finite $k$, continuous nonhydrodynamic spectra are encoded in the analytic continuation of $\mathfrak{D}_{IH}$, where they generically appear as branch cuts or other nonanalyticities in the complex $\omega$-plane.

\newpage
\noindent \textbf{\textit{Example 1: Massless RTA ---}} The resummed diffusivities can be evaluated analytically in the case of a massless gas of non-conserved bosons in the Relaxation-Time Approximation (RTA), with relaxation time $\tau_g$. The details of the calculation are in the Supplementary Material, while here we only report the result.

Along shear waves with gradients in direction $x$ and velocity fluctuations in direction $y$, the relevant balance law is $\partial_t \delta T^{0y}+\partial_x \delta T^{xy}=\text{``Source''}$. Given that, in local equilibrium, the shear stress $\delta T^{xy}$ vanishes, the gradient expansion \eqref{decomposco} reduces to $\delta T^{xy}=-ik\mathfrak{D}_{yy} \delta T^{0y}$ (we write $\mathfrak{D}_{yy}$ as a shorthand for $\mathfrak{D}_{T^{0y}T^{0y}}$), with resummed momentum diffusivity
\begin{equation}
\mathfrak{D}_{yy}(i\omega,ik)=  \dfrac{1-i\omega\tau_g}{(ik\tau_g)^2} \left[1-\frac{2/3 }{ \frac{(1-i\omega\tau_g)^2}{(ik\tau_g)^2}+\left(1-\frac{(1-i\omega\tau_g)^2}{(ik\tau_g)^2}\right) \frac{1-i\omega\tau_g}{ik\tau_g} \text{arctanh}\left(\frac{ik\tau_g}{1-i\omega\tau_g}\right)} \right] \tau_g \, .  
\end{equation}
For sound waves traveling in direction $x$, the relevant balance laws are $\partial_t \delta T^{00}+\partial_x \delta T^{x0}=\text{``Source''}$ and $\partial_t \delta T^{0x}+\partial_x \delta T^{xx}=\text{``Source''}$. The fluxes are $\delta T^{x0}=\delta T^{0x}$ and $\delta T^{xx}=\frac{1}{3}\delta T^{00}-ik\mathfrak{D}_{xx}\delta T^{0x}$, with resummed diffusivity 
\begin{equation}
\mathfrak{D}_{xx}(i\omega,ik)=\left[ \dfrac{(1-i\omega\tau_g)^2}{(ik\tau_g)^2}+\dfrac{ \frac{1}{3}\text{arctanh}(\frac{ik\tau_g}{1-i\omega\tau_g})}{ \frac{ik\tau_g}{1-i\omega\tau_g}-\text{arctanh}(\frac{ik\tau_g}{1-i\omega\tau_g})}\right] \dfrac{\tau_g}{1-i\omega\tau_g}\, .
\end{equation}

Setting $ik=0$, we obtain $\mathfrak{D}_{yy}=\frac{\tau_g}{5}(1-i\omega\tau_g)^{-1}$ and $\mathfrak{D}_{xx}=\frac{4\tau_g}{15}(1-i\omega\tau_g)^{-1}$, both of which have a pole at the location of non-hydrodynamic mode $i\omega=1/\tau_g$. If we further set $i\omega=0$, and compare with the hydrodynamic prediction, we recover the well-known formula for the shear viscosity of RTA: $\eta=\frac{1}{5}(\varepsilon+P)\tau_g$. Moreover, plotting the resummed diffusivities on the causal diamond (see Fig. \ref{fig:DiamondRTA}), one finds that they are both singular on the upper horizon, $i\omega=1/\tau_g -|ik|$ (note that $w=1$), due to the encounter with the ballistic spectrum, and cannot be continued further. Finally, plotting them on the complex $\omega-$plane for real $k$, we observe the usual RTA ballistic branch cut.

\begin{figure}[h!]
    \centering
\includegraphics[width=0.43\linewidth]{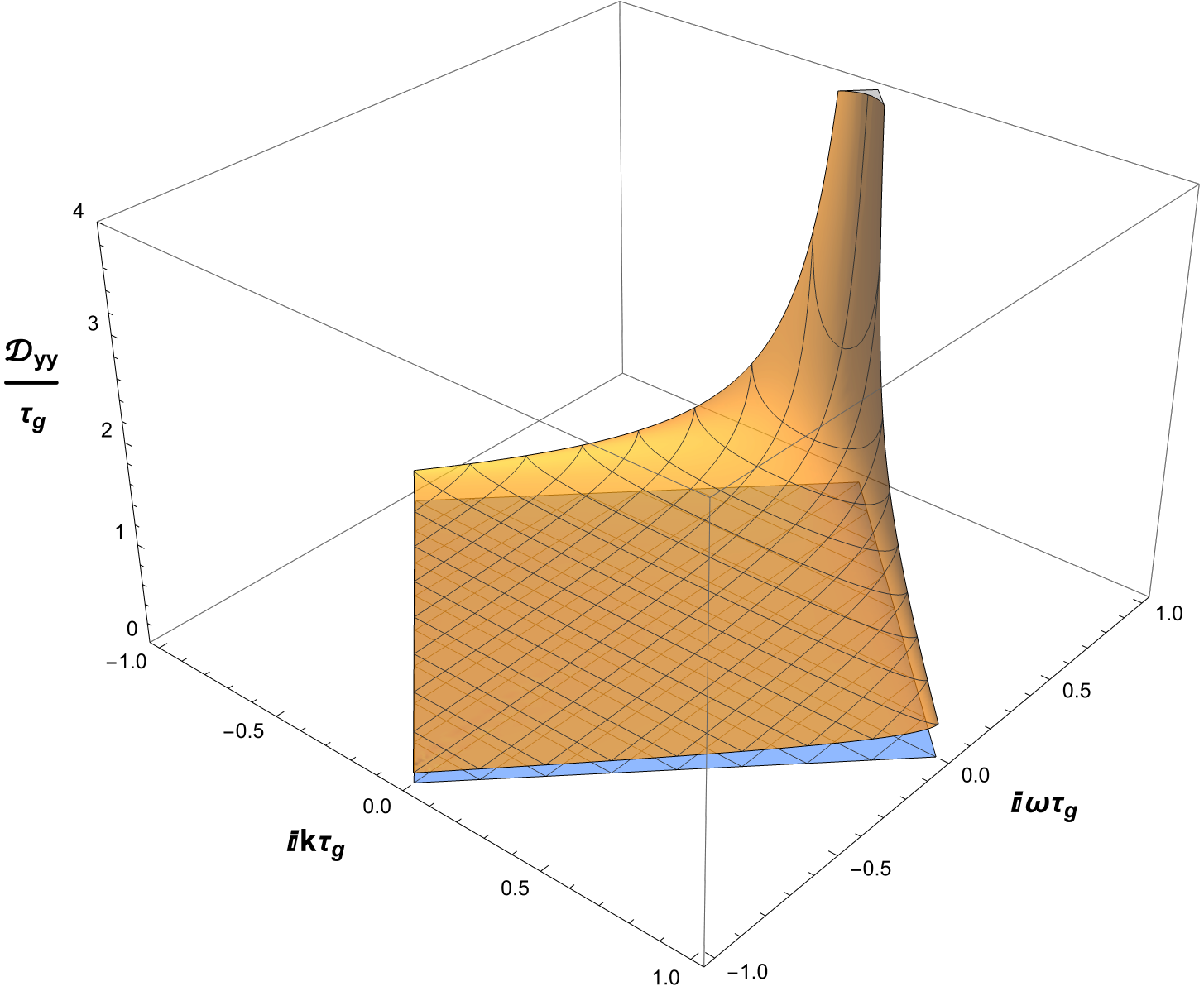}\hspace{0.02\linewidth}\includegraphics[width=0.43\linewidth]{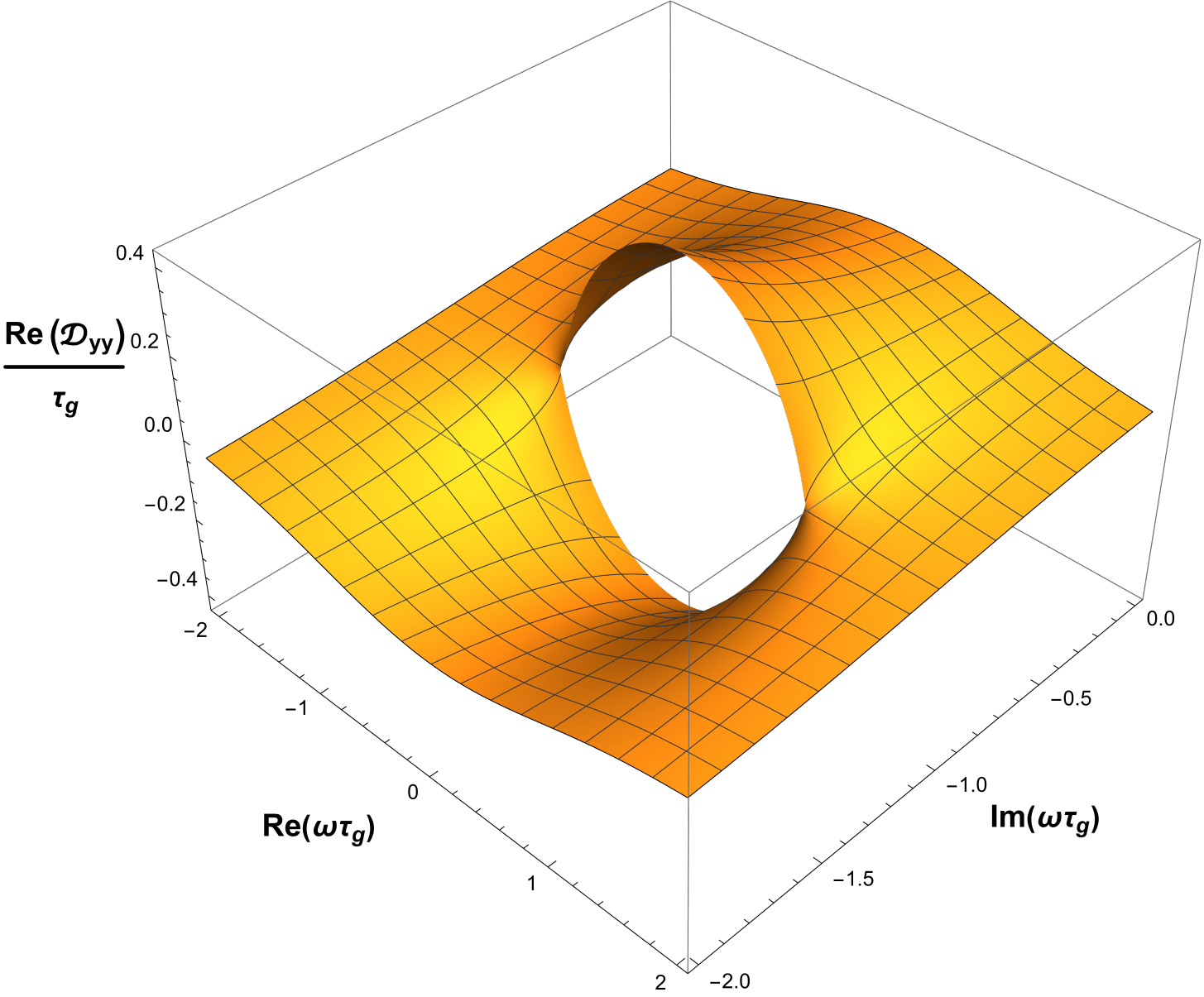}
\caption{Resummed diffusivity $\mathfrak{D}_{yy}$ of massless RTA kinetic theory on the causal diamond (left) and on the complex $\omega$-plane at fixed $k=1/\tau_g$ (right). \textit{Left panel}: On the relaxation plane, $\mathfrak{D}_{yy}$ remains real and non-negative throughout the diamond. This positivity is general: for any vector $V^I$, the quadratic form $V^{I*}\mathfrak{D}_{IH}V^H$ is a diagonal matrix element of the operator $[\s_N+ik\E_N-i\omega]^{-1}$, which is positive definite inside the diamond. Remarkably, the resummed diffusivity remains finite upon approaching the upper horizon (except at its tip), where $\mathfrak{D}_{yy}=1/|3ik|$, while its normal derivative diverges. Thus, the boundary of the convergence domain need not correspond to a divergence of the resummed diffusivity itself, but rather to a loss of regularity that prevents further analytic continuation. \textit{Right panel}: At fixed $k$, $\mathfrak{D}_{yy}$ exhibits a branch cut precisely along the nonhydrodynamic spectrum, $\omega=-i/\tau_g+[-k,k]$ \cite{RomatschkeCutsandPoles:2015gic}. The qualitative behavior of $\mathfrak{D}_{xx}$ is fully analogous.}
\label{fig:DiamondRTA}
\end{figure}

\noindent \textbf{\textit{Example 2: Fokker-Planck diffusion ---}} The resummed diffusivity can also be evaluated analytically for a massless gas of conserved particles traveling in one spatial dimension and interacting with an external medium in the Fokker-Planck approximation (infinitely frequent but infinitely soft scatterings) \cite{Debbasch,DunkelHanggi}. Denoting the non-resummed diffusivity by $\mathfrak{D}$, we find (see Supplementary Material) that the conservation law $\partial_t \delta n+\partial_x \delta J=\text{``Source''}$ is supplemented by the resummed constitutive relation $\delta J=-ik\mathfrak{D}_{nn}\delta n$, with
\begin{equation}\label{DnnFokkerPLanckwow}
\mathfrak{D}_{nn}(i\omega,ik)= \frac{ik \left(q_+-q_-\right)-i\omega \left(q_++q_--2\right)}{(i\omega)^2 \left(q_++q_-+2\right)-ik i\omega \left(q_+-q_-\right)-4 (ik)^2} \, , \qquad \qquad \text{with }q_\pm=\sqrt{1{-}4 \mathfrak{D} (i\omega{\pm}ik)} \, ,
\end{equation}
as is shown in figure \ref{fig:DiamondFokker}. Again, the singularities of the resummed diffusivity are due to ballistic branch cuts.

\begin{figure}[h!]
    \centering
\includegraphics[width=0.43\linewidth]{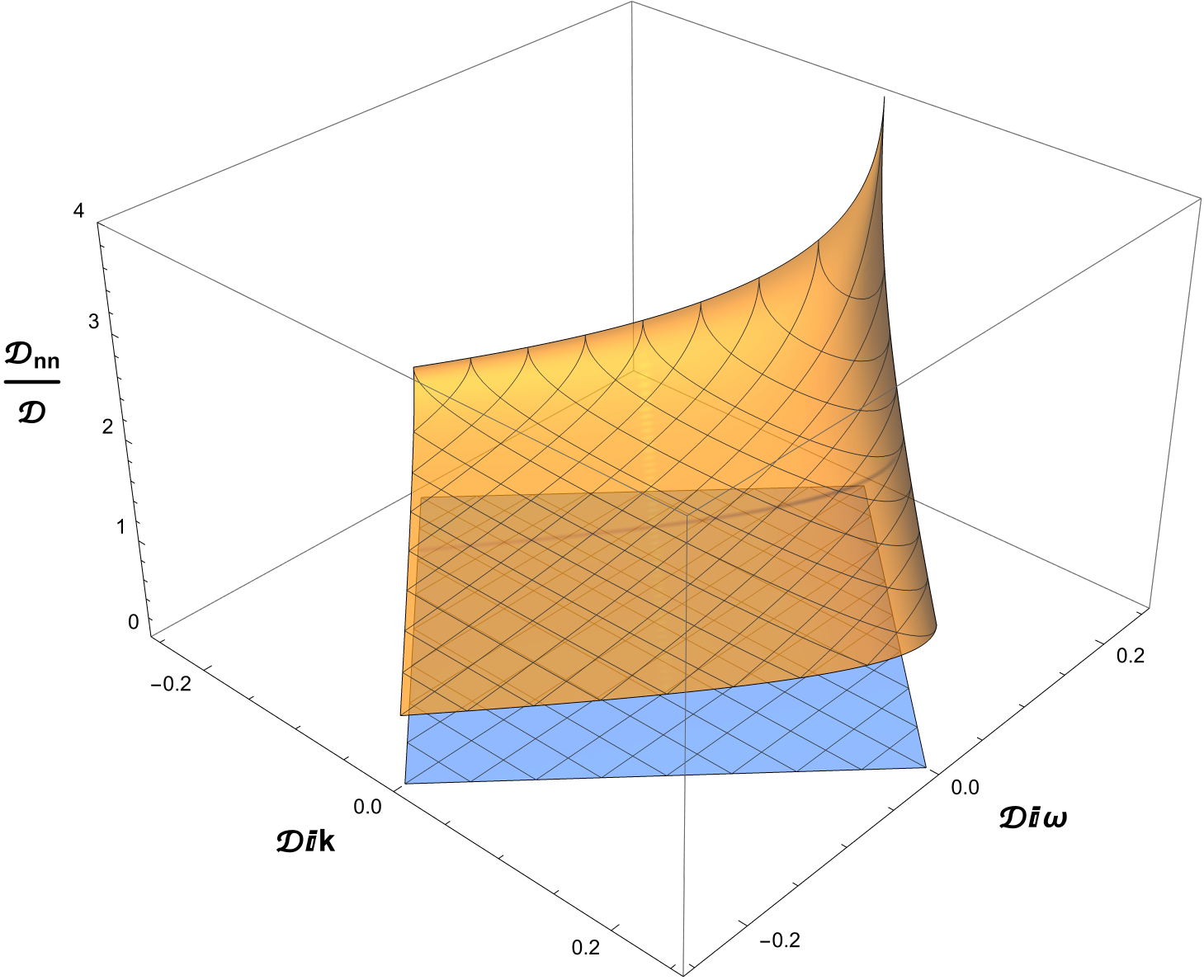}\hspace{0.02\linewidth}\includegraphics[width=0.43\linewidth]{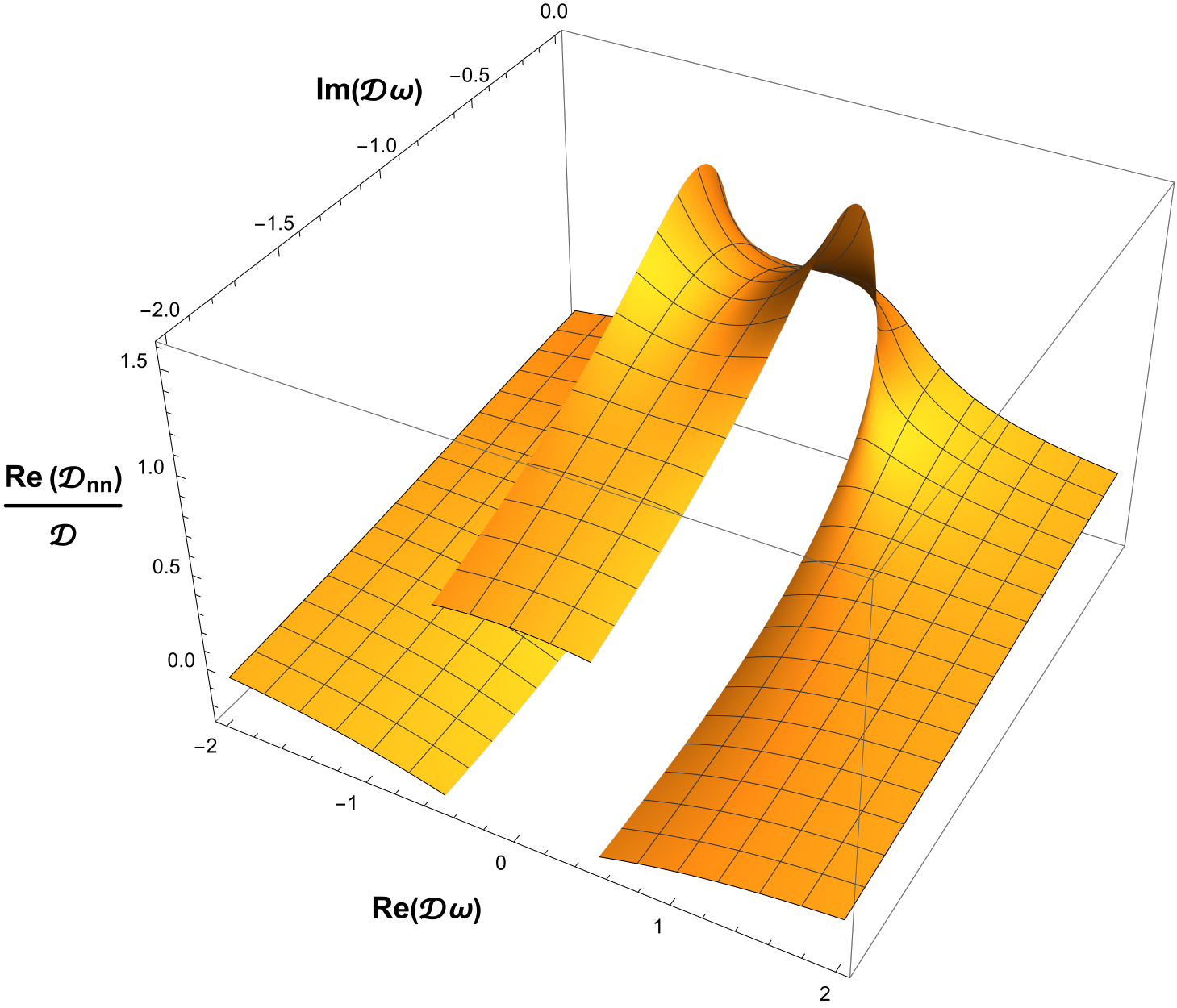}
\caption{Resummed diffusivity $\mathfrak{D}_{nn}$ of Fokker-Planck kinetic theory for massless particles confined in one spatial dimension. \textit{Left panel}: On the causal diamond $|i\omega|+|ik|<1/(4\mathfrak{D})$, the behavior is similar to that found in RTA (see figure \ref{fig:DiamondRTA}). The resummed diffusivity remains finite along the upper horizon (it is also finite on the tip), while its normal derivative diverges. \textit{Right panel}: At fixed $k$ (here, $k\,{=}\,(2\mathfrak{D})^{-1}$), $\mathfrak{D}_{nn}$ exhibits two branch cuts along the nonhydrodynamic spectrum, $\omega=\pm k-i[\frac{1}{4\mathfrak{D}},\infty)$ \cite{GavassinoFokkerPLanck:2026zsz}. On the relaxation plane, the corresponding branch points trace precisely the upper horizon of the causal diamond.}
\label{fig:DiamondFokker}
\end{figure}

\newpage
The two examples above admit a substantial generalization. In a companion paper, we compute the resummed diffusivity exactly for monoenergetic particles undergoing arbitrary rotationally invariant elastic scattering in three spatial dimensions. The result takes the form of an infinite continued fraction, revealing explicitly how the all-order constitutive relation encodes the angular structure of the microscopic scattering process.

\noindent \textbf{\textit{Off-shell transport coefficients ---}} If we compare the gradient expansion \eqref{JIgeneral}, which is a power series in two variables, with the Taylor expansion $\omega(k)=\sum_n a_n k^n$ of the hydrodynamic modes, which is a power series in one variable, we see that, at any fixed order, there are many more coefficients in the former. This is not a redundancy of \eqref{JIgeneral}. Rather, two physical systems with identical hydrodynamic dispersion relations can respond differently to an external hydrodynamic force.
Take, for instance, a conserved particle density $\delta n$ whose source-free evolution is governed by $i\omega=-\mathfrak{D}(ik)^2$. This dispersion relation follows from Fick's law, $\delta J=-ik\mathfrak{D}\delta n$, which contains a single transport coefficient. However, exactly the same dispersion relation follows from $\delta J=-ik\mathfrak{D}_{nn}(i\omega,ik)\delta n$, with $\mathfrak{D}_{nn}$ given by the Fokker-Planck law \eqref{DnnFokkerPLanckwow} \cite{GavassinoFokkerPLanck:2026zsz}, which contains infinitely many coefficients. Thus, Fick and Fokker-Planck are indistinguishable along source-free solutions \cite{GavassinoDiffusionCompatible:2026tvy}, but respond differently to external forces.

This motivates one to call \textit{off-shell transport coefficients} those coefficients that cannot be determined individually from source-free hydrodynamic evolution, but can be disentangled through external forcing \cite{Haehl:2015pja,Bu:2015ame,Bu:2018drd,Diles:2023tau}. For particle diffusion, writing the hydrodynamic mode as $i\omega=-\mathfrak{D}(ik)^2-\mathfrak{B}(ik)^4+\ldots$, where $\mathfrak{D}$ is the ordinary diffusivity and $\mathfrak{B}$ the super-Burnett coefficient, we can reorganize the constitutive relation as
\vspace{-0.2cm}
\begin{equation}
\begin{split}
\mathfrak{D}_{nn}(i\omega,ik) &=\mathfrak{D}+\mathfrak{D}_{nn10}i\omega+\mathfrak{D}_{nn02}(ik)^2+\ldots\\
&=\mathfrak{D}+\mathfrak{D}_{nn10}[i\omega{+}\mathfrak{D}(ik)^2+\ldots]+[\mathfrak{D}_{nn02}{-}\mathfrak{D}_{nn10}\mathfrak{D}](ik)^2+\ldots
\xrightarrow{\text{on-shell}}
\mathfrak{D}+\mathfrak{B}(ik)^2+\ldots \, .
\end{split}
\end{equation}
Hence, while the combination $\mathfrak{B}=\mathfrak{D}_{nn02}-\mathfrak{D}_{nn10}\mathfrak{D}$ is measurable from source-free evolution, the individual coefficients $\mathfrak{D}_{nn10}$ and $\mathfrak{D}_{nn02}$ cannot be disentangled without external forcing. Fokker-Planck theory is particularly striking in this respect: it has infinitely many off-shell transport coefficients, yet they all cancel on-shell, leaving only the diffusivity. For instance, we have that $\mathfrak{D}_{nn02}=2\mathfrak{D}^3$ and $\mathfrak{D}_{nn10}=2\mathfrak{D}^2$, giving $\mathfrak{B}\equiv0$.

\noindent \textbf{\textit{Conclusions ---}} We have argued that the convergence of hydrodynamics is most naturally formulated off-shell, at the level of constitutive relations, rather than individual dispersion relations. The latter may become nonanalytic arbitrarily close to the origin for reasons unrelated to any breakdown of \eqref{JIgeneral}. By contrast, causality and the nonhydrodynamic gap guarantee a universal convergence domain for the resummed transport coefficients in \eqref{decomposco}, which takes the form of a causal diamond on the relaxation plane.

This viewpoint also exposes information that is lost on-shell. The singularities encountered by analytically continuing the all-order gradient expansion encode nonhydrodynamic information, while external forcing gives access to transport coefficients that cannot be disentangled from source-free evolution. Thus, the gradient expansion contains considerably more microscopic information than is visible in hydrodynamic dispersion relations. In this sense, off-shell hydrodynamics knows about the degrees of freedom beyond hydrodynamics.

\noindent \textbf{\textit{Note added. }} After completion of this manuscript, we learned of independent results by A. Serantes on a related convergence problem. Their analysis concerns finite-dimensional systems and allows for oscillatory nonhydrodynamic degrees of freedom, while the present result accommodates infinite-dimensional systems but assumes a purely relaxational nonhydrodynamic spectrum. We thank A. Serantes for discussions and for bringing these results to our attention.

\section*{Acknowledgements}
\vspace{-0.4cm}

This work is supported by a MERAC Foundation prize grant,  an Isaac Newton Trust Grant, and funding from the Cambridge Centre for Theoretical Cosmology. I thank M. Spalinski and A. Serantes for reading the manuscript and providing useful feedback.
\vspace{-0.5cm}

\bibliography{Biblio}

\newpage

\onecolumngrid
\newpage
\begin{center}
\textbf{\large The causal diamond of relativistic hydrodynamics:\\Supplementary Material}\\[.2cm]
  L. Gavassino\\[.1cm]
  {\itshape Department of Applied Mathematics and Theoretical Physics, University of Cambridge, Wilberforce Road, Cambridge CB3 0WA, United Kingdom\\}
\end{center}

\setcounter{equation}{0}
\setcounter{figure}{0}
\setcounter{table}{0}
\setcounter{page}{1}
\renewcommand{\theequation}{S\arabic{equation}}
\renewcommand{\thefigure}{S\arabic{figure}}

\section{Sound-sound collision in a kinetic model}
\vspace{-0.3cm}

Here, we present a kinetic toy model of viscous sound in which we can continuously send the sound speed to zero while keeping $\mathfrak{D}_{IH}(i\omega,ik)$ fixed. As we shall see, also in this case a sound-sound collision takes place at $k_c\tau_g= \mathcal{O}(c_s)$, while the singularities of $\mathfrak{D}_{IH}$ (which are branch points) lie at a fixed, finite distance from the origin in $(i\omega,ik)$ space. Unlike in the Israel-Stewart example, the constitutive relation in this model contains infinitely many non-vanishing higher-order corrections in $k$.

We consider the following system:
\begin{equation}\label{kineticIS}
\begin{split}
&\partial_t \delta \varepsilon+\partial_x \delta\mathcal{F}=0 \, , \\
&\partial_t \delta \mathcal{F}+c_s^2 \partial_x \delta\varepsilon+c_b\partial_x \langle \delta I n_x^2\rangle=0 \, , \\
&(\partial_t +c_0n^x \partial_x)\delta I+\tau_g^{-1} \delta I +c_b n_x^2 \partial_x \delta\mathcal{F}=0 \, . \\
\end{split}
\end{equation}
The first line is the conservation of energy, while the second is the conservation of momentum. The function $\delta I(n^j)\in L^2(S^2)$ is an angular distribution that parameterizes the stress anisotropy, with $\langle \cdot\rangle$ the average over all directions. The third line is a Boltzmann-like equation for $\delta I$. Note that the structure of \eqref{kineticIS} is a generalization of Israel-Stewart, where instead of having a single $\delta\Pi$ that relaxes locally, we have a continuous angular distribution that is transported from point to point. The coupling factors are chosen so that the second law holds, with free-energy density $2(\varepsilon+P)\Delta\mathfrak{F}=c_s^2 (\delta\varepsilon)^2 +(\delta\mathcal{F})^2+\langle (\delta I)^2\rangle$.

Similarly to the Israel-Stewart case, the conserved densities are $\rho_I=(\delta \varepsilon,\delta\mathcal{F})$, and the only non-vanishing diffusivity is $\mathfrak{D}_{\mathcal{F}\mathcal{F}}$. Working in Fourier space, we can solve for $\delta I$, and obtain
\begin{equation}
\mathfrak{D}_{\mathcal{F}\mathcal{F}}=\left\langle \dfrac{n_x^4 c_b^2 \tau_g}{1-i\omega\tau_g +c_0 n^x ik\tau_g} \right\rangle =\sum_{a=0}^\infty \left\langle n_x^4 c_b^2 \tau_g (i\omega\tau_g -c_0 n^x ik\tau_g)^a \right\rangle\, ,
\end{equation}
where the series is guaranteed to converge for $|i\omega|+|c_0 ik| < 1/\tau_g$. 

At the same time, the conservation laws imply $(i\omega)^2 +i\omega (ik)^2 \mathfrak{D}_{\mathcal{F}\mathcal{F}}-c_s^2 (ik)^2=0$. Assuming that $i\omega$ and $k$ are real, this equation can be solved analytically, giving two parametric curves, which are plotted in Fig. \ref{fig:KineticToy}. We see that these curves bifurcate from a small finite momentum $k_c$, meaning that the sound waves cease to be propagating at that wavenumber. To estimate the value of $k_c$ we simply notice that, for $|\omega\tau_g|,|k\tau_g| \ll 1$, we can approximate $\mathfrak{D}_{\mathcal{F}\mathcal{F}}$ with its Navier-Stokes value $c_b^2 \tau_g/5$, which results in a predicted collision at $k_c\tau_g\simeq 10c_s/c_b^2$. Again, in the limit of small $c_s$, this collision happens arbitrarily close to the origin. In contrast, $\mathfrak{D}_{\mathcal{F}\mathcal{F}}$ is independent of $c_s$, and its branch points remain at a fixed, finite distance from the origin.

\begin{figure}[h!]
    \centering
\includegraphics[width=0.42\linewidth]{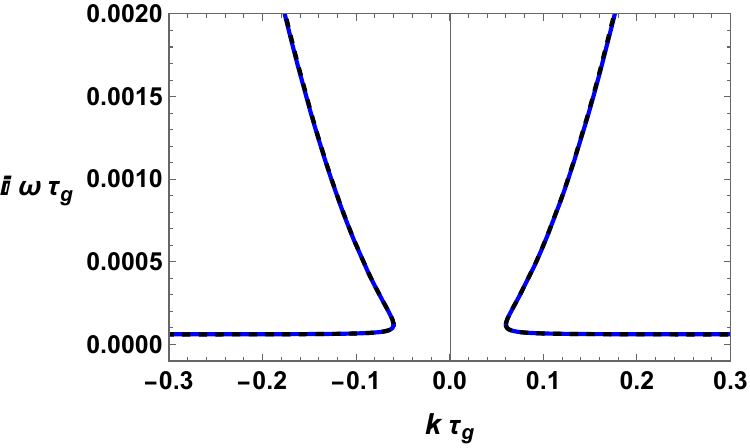}
\caption{Overdamped sound modes of the kinetic model \eqref{kineticIS}, for $c_s=0.002$ and $c_b=c_0=1/\sqrt{3}$. The solid blue curves are the exact result, while the dashed curves are obtained using the Navier-Stokes approximation $\mathfrak{D}_{\mathcal{F}\mathcal{F}}\simeq c_b^2\tau_g/5$. The curves bifurcate at $|k|\tau_g\equiv k_c \tau_g\simeq0.06$, corresponding to the collision of the two sound modes. For $|k|<k_c$, the modes form an underdamped pair and therefore do not lie on the real-$i\omega$ plane.}
    \label{fig:KineticToy}
\end{figure}

\newpage
\section{Regularity of the gradient expansion on the Fourier plane}

We have the following result:
\begin{theorem}\label{theoFourier}
All functions $\mathfrak{D}_{IH}(i\omega,ik)$ are holomorphic in a complex open neighborhood of any point $(\omega_0,k_0)\in\mathbb{R}^2$.
\end{theorem}
\begin{proof}
The operators $\s_N+ik\E_N$ form a holomorphic family
\cite[\S 7.1.2, Problem 1.2]{Kato_Perturbation_Theory}. Thus, it suffices to show that $i\omega_0$ belongs to the resolvent set of $\s_N+ik_0\E_N$, since in that case $(\s_N+ik\E_N-i\omega)^{-1}$ is bounded-holomorphic in the two complex variables $(i\omega,ik)$ in a neighborhood of $(i\omega_0,ik_0)$, and thus $\mathfrak{D}_{IH}(i\omega,ik)$, being matrix elements of this operator family, are in turn holomorphic there.

Let $\G\equiv\s_N+ik_0\E_N-i\omega_0$. We first show that $\G$ is bounded below. For any $\Psi$ in its domain, we have $\mathfrak{Re}(\Psi,\G\Psi)=(\Psi,\s_N\Psi)\geq ||\Psi||^2/\tau_g$. Hence, by the Cauchy-Schwarz inequality, $||\G\Psi||\geq ||\Psi||/\tau_g$. In particular, $\G$ is injective. Moreover, since $\G$ is closed, this lower bound implies that $\operatorname{Ran}(\G)$ is closed.

It remains to prove that $\G$ is surjective. Since $\G^\dagger=\s_N-ik_0\E_N+i\omega_0$, the same argument gives $\ker(\G^\dagger)=\{0\}$. Therefore, $\overline{\operatorname{ran}(\G)}=[\ker(\G^\dagger)]^\perp=\mathcal{H}_N$. Since $\operatorname{ran}(\G)$ is also closed, it follows that $\operatorname{ran}(\G)=\mathcal{H}_N$. Thus, $\G$ is bijective, and its inverse is bounded, with $||\G^{-1}||\leq\tau_g$. Hence, $i\omega_0$ belongs to the resolvent set of $\s_N+ik_0\E_N$, completing the proof.
\end{proof}

\newpage
\section{Gradient expansion of RTA}

\subsubsection{Outline of the theory}

We consider a gas of massless non-conserved bosons in 3+1 dimensions. Assuming that gradients always point in direction $x$, we can characterize linear fluctuations using perturbations to the distribution function $f(t,x,p^j)$, which counts how many bosons occupy a single-particle state located at $(t,x)$ and with momentum $p^j$. The linearized Boltzmann equation in the Relaxation Time Approximation (RTA) reads
\begin{equation}\label{RTAeom}
(\partial_t +v^x \partial_x)\delta f= \dfrac{1}{\tau_g} [f_{BB}(1+f_{BB})p^\alpha\delta\beta_\alpha-\delta f] \, .
\end{equation}
Here, $v^j=p^j/p^0\in S^2$ is the particle velocity, and $f_{BB}=(e^{p^0/T}-1)^{-1}$ is the equilibrium black-body distribution of the gas in the unperturbed state. The relaxation time $\tau_g$ is assumed to be constant. The vector $\delta\beta_\alpha$ is determined through Landau matching:
\begin{equation}\label{RTAmatching}
\int d^3 p \, p^\nu[f_{BB}(1+f_{BB})p^\alpha\delta\beta_\alpha-\delta f]=0 \, .
\end{equation}
This guarantees that the conservation laws $\partial_t \delta T^{0\nu}+\partial_x\delta T^{x\nu}=0$ are respected, where the energy-momentum tensor is given by
\begin{equation}\label{energymomentumtensor}
\delta T^{\mu \nu}=\int \dfrac{gd^3 p}{(2\pi)^3p^0} p^\mu p^\nu \delta f \qquad\qquad (g=\text{spin degeneracy}) \, .
\end{equation}

We will also need an expression for the quadratic free-energy density associated to $\delta f$:
\begin{equation}\label{deltrafreenergy}
2\Delta\mathfrak{F}=\int \dfrac{gd^3 p}{(2\pi)^3} \dfrac{T(\delta f)^2}{f_{BB}(1+f_{BB})} \, .
\end{equation}

\subsubsection{Shear channel}

In order to isolate the shear channel, and remove some degrees of freedom that are irrelevant for hydrodynamics, we make the following ansatz:
\begin{equation}\label{ansatzshear}
\delta f(p^j)=\sqrt{\dfrac{4}{c_v T^3}} f_{BB}(1+f_{BB})p^0 \dfrac{v_y \Psi(v^x)}{\sqrt{1-v_x^2}} \qquad \qquad (c_v=\text{specific heat}) \, .
\end{equation}
This construction selects excitations that have spin $1$ (relative to rotations around the $x$-axis), and integrates out the energy dependence\footnote{Note that the ansatz \eqref{ansatzshear} is introduced solely for convenience. In general, one can keep $\delta f$ arbitrary, and rescale it as $\delta f=\sqrt{f(1{+}f)}\Psi$. With this construction, one finds that $\mathcal{H}=L^2(\mathbb{R}^3)$, and the natural Onsager inner product is $(\Phi,\Psi)=\int_{\mathbb{R}^3} \Phi^*\Psi d^3 p$ \cite{DudynskiEkielJezewska1985}. Our choice, instead, is convenient because it allows us to work in a small linear subspace of $L^2(\mathbb{R}^3)$, which, as we showed, is isomorphic to $L^2([-1,1])$. The reason why this could be done is that the ansatz \eqref{ansatzshear} selects a linear subspace of $L^2(\mathbb{R}^3)$ that happens to be invariant under the action of both $\s$ and $\E$ (viewed as operators on the larger space). In a more complicated kinetic theory, this kind of simplification may not be available.}. Our choice of factors guarantees that the natural Onsager inner product, defined via \eqref{deltrafreenergy}, is just $(\Phi,\Psi)=\int_{-1}^1 \Phi^* \Psi dv_x$, and thus we can set $\mathcal{H}=L^2([-1,1])$.

Plugging \eqref{ansatzshear} into \eqref{RTAeom} and \eqref{RTAmatching}, and introducing the normalized vector $e_y=\sqrt{1{-}v_x^2}/||\sqrt{1{-}v_x^2}||$ we obtain
\begin{equation}
(\partial_t+v_x \partial_x)\Psi= \dfrac{1}{\tau_g} [e_y(e_y,\Psi)-\Psi] \, .
\end{equation}
From here, we immediately read off the transport operator $\E=v_x$, the relaxation operator $\s=\tau_g^{-1}[1-e_y(e_y,\cdot)]$, and the conserved space $\mathcal{H}_C=\ker(\s)=\text{span}\{e_y\}$. Thus, there is a single relevant conservation law in this channel, with density $\rho_y=(e_y,\Psi)$ and flux $J_y=(e_y,v_x\Psi)$, whose local equilibrium part is $J_y^{l.e.}=\chi_{yy} \rho_y=(e_y,v_xe_y)\rho_y=0$. The corresponding conserved charge is the momentum in direction $y$, since one finds from \eqref{energymomentumtensor} that $\delta T^{0y}=\sqrt{c_vT/3}\,(e_y,\Psi)$ and $\delta T^{xy}=\sqrt{c_vT/3}\,(e_y,v_x\Psi)$. Hence, according to the discussion in the main text, the quantity
\begin{equation}\label{diffusivitySheatRTA}
\mathfrak{D}_{yy}(i\omega,ik)= (\PN v_x e_y,[\s_N+ik\E_N-i\omega]^{-1}\PN v_x e_y)
\end{equation}
is the resummed momentum diffusivity of shear waves. In the following, we shall evaluate it explicitly.

By definition, $\PN=1-e_y(e_y,\cdot)$, so that $\PN v_x e_y=v_x e_y-e_y(e_y,v_xe_y)=v_x e_y$. Moreover, if $\Psi \in \mathcal{H}_N=e_y^\perp$, the $\s\Psi=\tau_g^{-1}$, so that $\s_N=1/\tau_g$. Therefore, equation \eqref{diffusivitySheatRTA} reduces to
\begin{equation}\label{diffusivitySheatRTA}
\mathfrak{D}_{yy}(i\omega,ik)= \dfrac{\tau_g}{1-i\omega\tau_g} (v_x e_y,[1+r \E_N]^{-1}v_x e_y) \,, \qquad\qquad r=\dfrac{ik\tau_g}{1-i\omega\tau_g} \, .
\end{equation}
Now, we recall that the compression $\E_N$ acts as follows: $\E_N \Psi =v_x \Psi-e_y (e_y,v_x\Psi)$ (for $\Psi \in \mathcal{H}_N$). Therefore, to determine $[1+r\E_N]^{-1}$, we need to invert the relation $\Psi=(1+r\E_N)\Phi=(1+rv_x)\Phi- re_y(e_y,v_x\Phi)$, under the assumption that $\Psi,\Phi \perp e_y$. With a bit of algebra, one obtains $\Phi=(1+r\E_N)^{-1}\Psi$, with
\begin{equation}
(1+r\E_N)^{-1}\Psi=\dfrac{\Psi}{1{+}rv_x} -\dfrac{\dfrac{e_y}{1{+}rv_x} \left(e_y,\dfrac{\Psi}{1{+}rv_x} \right)}{\left(e_y,\dfrac{e_y}{1{+}rv_x}\right)} \, .
\end{equation}
Plugging this into \eqref{diffusivitySheatRTA}, and evaluating the inner products explicitly, we obtain
\begin{equation}
\mathfrak{D}_{yy}(i\omega,ik)=  \dfrac{1}{r^2} \left[1-\frac{2/3 }{ \frac{1}{r^2}+\left(1-\frac{1}{r^2}\right) \frac{\text{arctanh}\left(r\right)}{r}} \right] \dfrac{\tau_g}{1-i\omega\tau_g} \, ,
\end{equation}
where we recall that $r=ik\tau_g/(1-i\omega\tau_g)$.
As a consistency check, one can verify explicitly that the conservation law $i\omega=-\mathfrak{D}_{yy} (ik)^2$ gives the usual dispersion relation of shear waves in RTA \cite{RomatschkeCutsandPoles:2015gic,GavassinoNonNewtonianRadiation:2024jej}.

\subsubsection{Sound channel}

To isolate the sound channel (i.e. spin $0$ excitations) we make the ansatz
\begin{equation}\label{ansatzsound}
\delta f(p^j)=\sqrt{\dfrac{2}{c_vT^3}} f_{BB}(1+f_{BB})p^0 \Psi(v^x) \, .
\end{equation}
With this choice of normalization, the natural Onsager inner product again reads $(\Phi,\Psi)=\int_{-1}^1 \Phi^* \Psi dv_x$, and so again $\mathcal{H}=L^2([-1,1])$. Plugging \eqref{ansatzsound} into \eqref{RTAeom} and \eqref{RTAmatching}, and defining $e_t=1/||1||$ and $e_x=v_x/||v_x||$, we obtain
\begin{equation}
(\partial_t +v_x \partial_x)\Psi= \dfrac{1}{\tau_g}[e_t(e_t,\Psi)+e_x(e_x,\Psi)-\Psi] \, ,
\end{equation}
from which we read off $\E=v_x$, $\s=\tau_g^{-1}[1-e_t(e_t,\cdot)-e_x(e_x,\cdot)]$, and $\mathcal{H}_C=\ker(\s)=\text{span}\{e_t,e_x\}$. Therefore, there are two conserved densities, $\rho_t=(e_t,\Psi)$ and $\rho_x=(e_x,\Psi)$, with corresponding fluxes $J_t=(e_t,v_x\Psi)$ and $J_x=(e_x,v_x\Psi)$. Using \eqref{energymomentumtensor}, we find that $\delta T^{tt}=\sqrt{c_v T} \rho_t$ and $\delta T^{tx}=\sqrt{c_v T} J_t$, so that the first conservation law is just the conservation of energy, while $\delta T^{tx}=\sqrt{c_v T/3} \rho_x$ and $\delta T^{xx}=\sqrt{c_v T/3} J_x$, so that the second conservation law is the conservation of momentum in direction $x$. Using the fact that $v_xe_t=\frac{1}{\sqrt{3}} e_x$, we immediately find that
\begin{equation}
\begin{bmatrix}
J_t \\
J_x \\
\end{bmatrix}
=
\underbrace{\begin{bmatrix}
0 & 1/\sqrt{3} \\
1/\sqrt{3} & 0 \\
\end{bmatrix}
\begin{bmatrix}
\rho_t \\
\rho_x \\
\end{bmatrix}}_{\text{Local equilibrium}}-\underbrace{ik
\begin{bmatrix}
0 & 0 \\
0 & \mathfrak{D}_{xx} \\
\end{bmatrix}
\begin{bmatrix}
\rho_t \\
\rho_x \\
\end{bmatrix}}_{\text{Diffusive part}} \, .
\end{equation}
The local equilibrium term describes ideal fluid dynamics, with speed of sound $1/\sqrt{3}$. The diffusive part describes viscosity. Our task is to compute the matrix element $\mathfrak{D}_{xx}(i\omega,ik)=(\PN v_x e_x,[\s_N+ik\E_N-i\omega]^{-1}\PN v_x e_x)$. Using the fact that $\PN v_x e_x=v_x e_x{-}\tfrac{1}{\sqrt{3}}e_t=\sqrt{3} (v_x^2-\tfrac{1}{3})e_t$, and $\s_N=1/\tau_g$, we find that
\begin{equation}
\mathfrak{D}_{xx}(i\omega,ik)=\dfrac{3\tau_g}{1-i\omega\tau_g}\left( (v_x^2-\tfrac{1}{3})e_t,[1+r\E_N]^{-1} (v_x^2-\tfrac{1}{3})e_t\right) \,, \qquad\qquad r=\dfrac{ik\tau_g}{1-i\omega\tau_g} \, .
\end{equation}
Additionally, inverting the relation $\Psi=(1+r\E_N)\Phi$, with $\Psi,\Phi\perp e_t,e_x$, we obtain
\begin{equation}
(1+r\E_N)^{-1}\Psi=\dfrac{\Psi}{1{+}rv_x} -\dfrac{\dfrac{e_x}{1{+}rv_x} \left(e_t,\dfrac{\Psi}{1{+}rv_x} \right)}{\left(e_t, \dfrac{e_x}{1{+}rv_x}\right)} \, .
\end{equation}
Hence, we finally obtain
\begin{equation}
\mathfrak{D}_{xx}(i\omega,ik)=\left[ \dfrac{1}{r^2}+\dfrac{ \text{arctanh}(r)/3}{ r-\text{arctanh}(r)}\right] \dfrac{\tau_g}{1-i\omega\tau_g} \, .
\end{equation}
Also in this case, if one solves the conservation laws (with no source) using this resummed diffusivity, one recovers the dispersion relation of sound waves in RTA \cite{GavassinoMovingInterfaces:2026fkj}.

\subsubsection{Behavior of the resummed diffusivities in the Fourier plane}


\begin{figure}[h!]
    \centering
\includegraphics[width=0.48\linewidth]{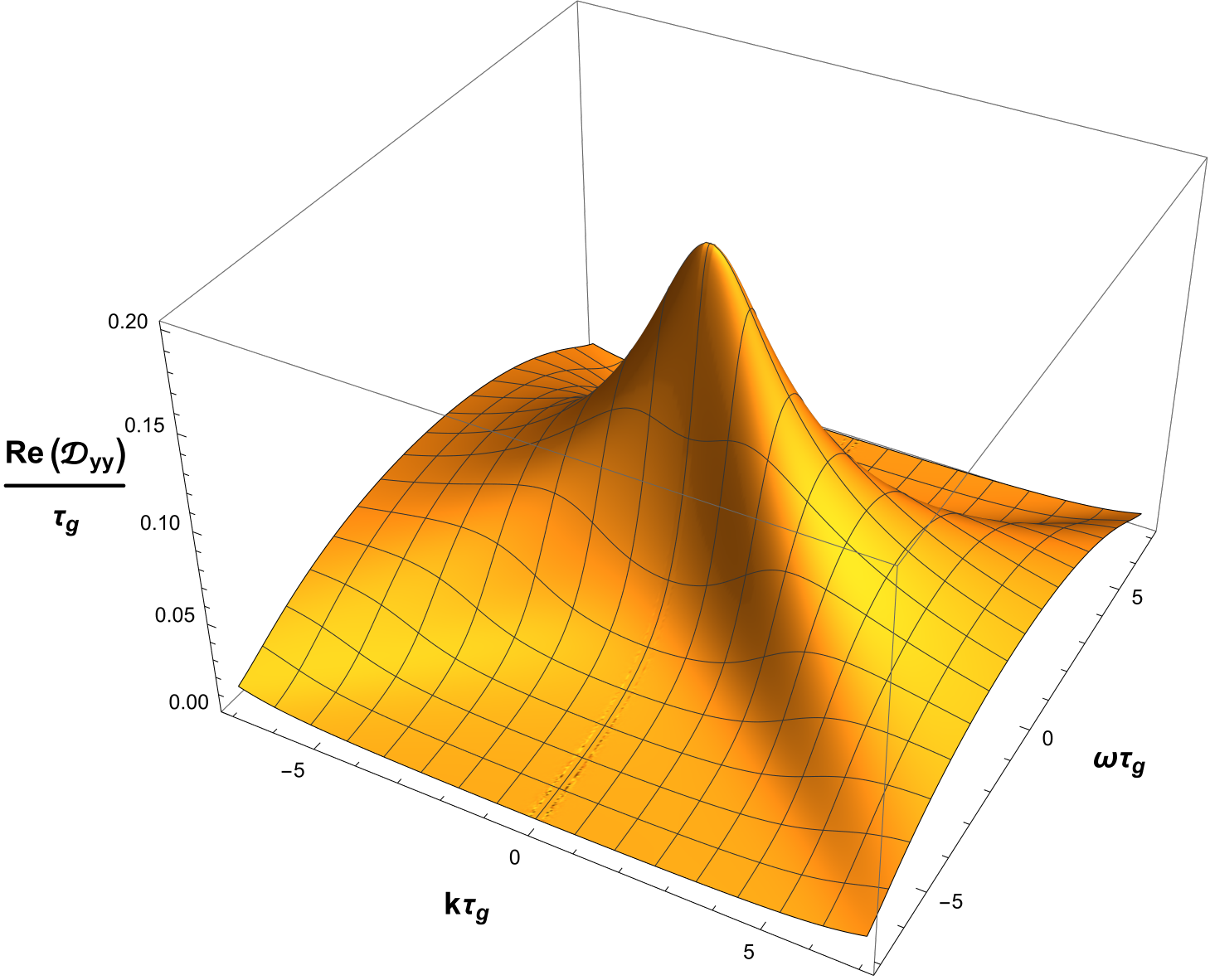}\hspace{0.02\linewidth}\includegraphics[width=0.48\linewidth]{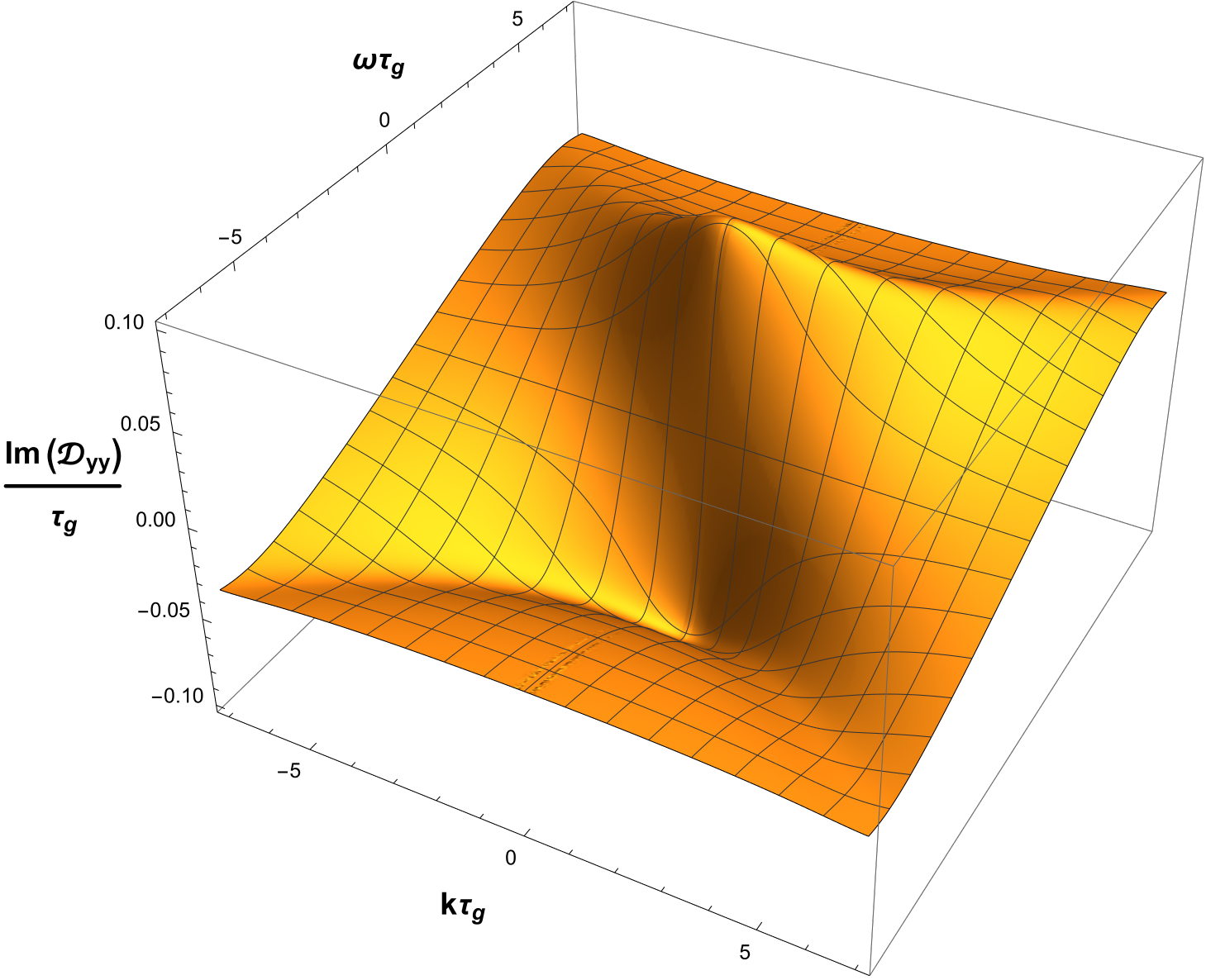}
\includegraphics[width=0.48\linewidth]{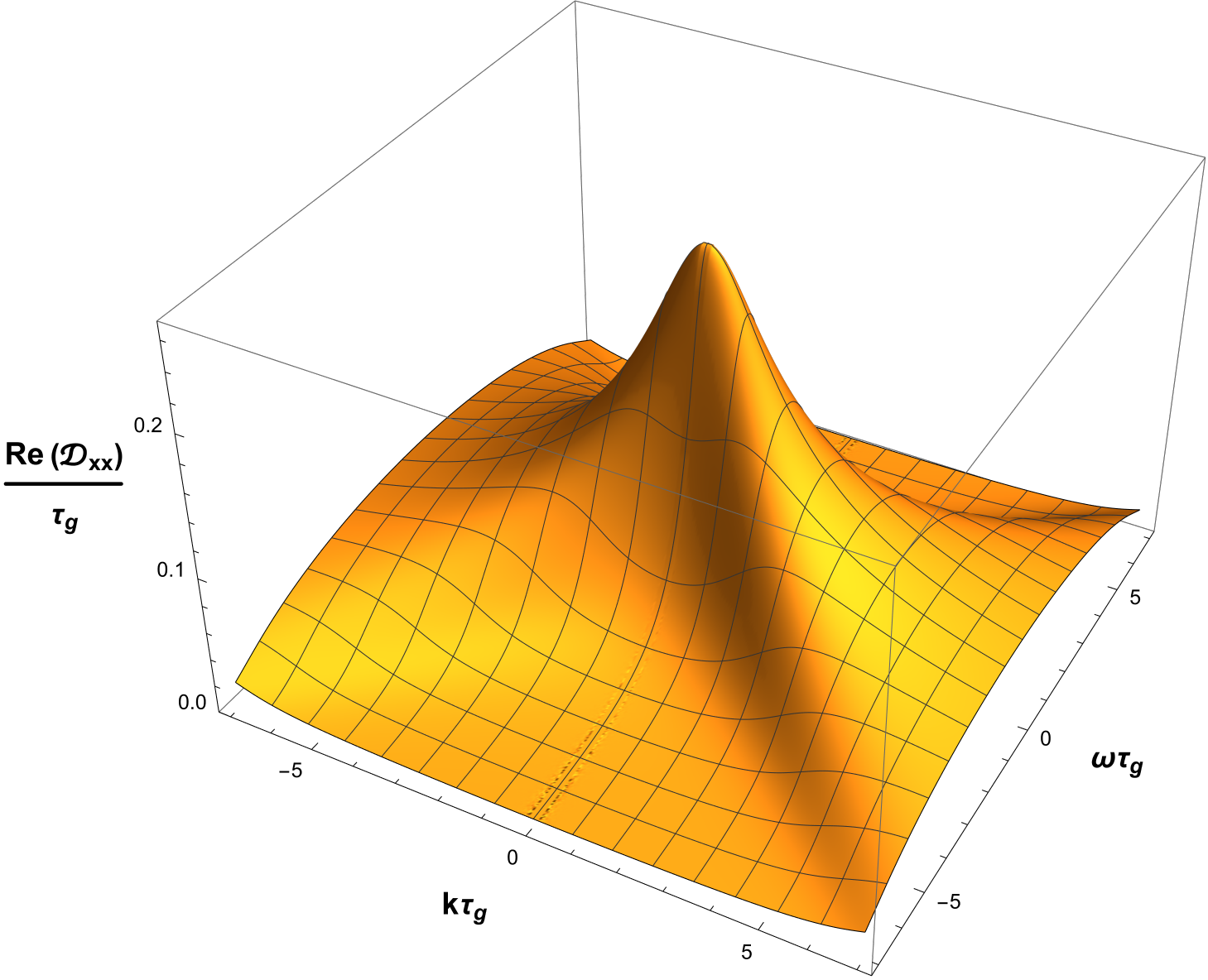}\hspace{0.02\linewidth}\includegraphics[width=0.48\linewidth]{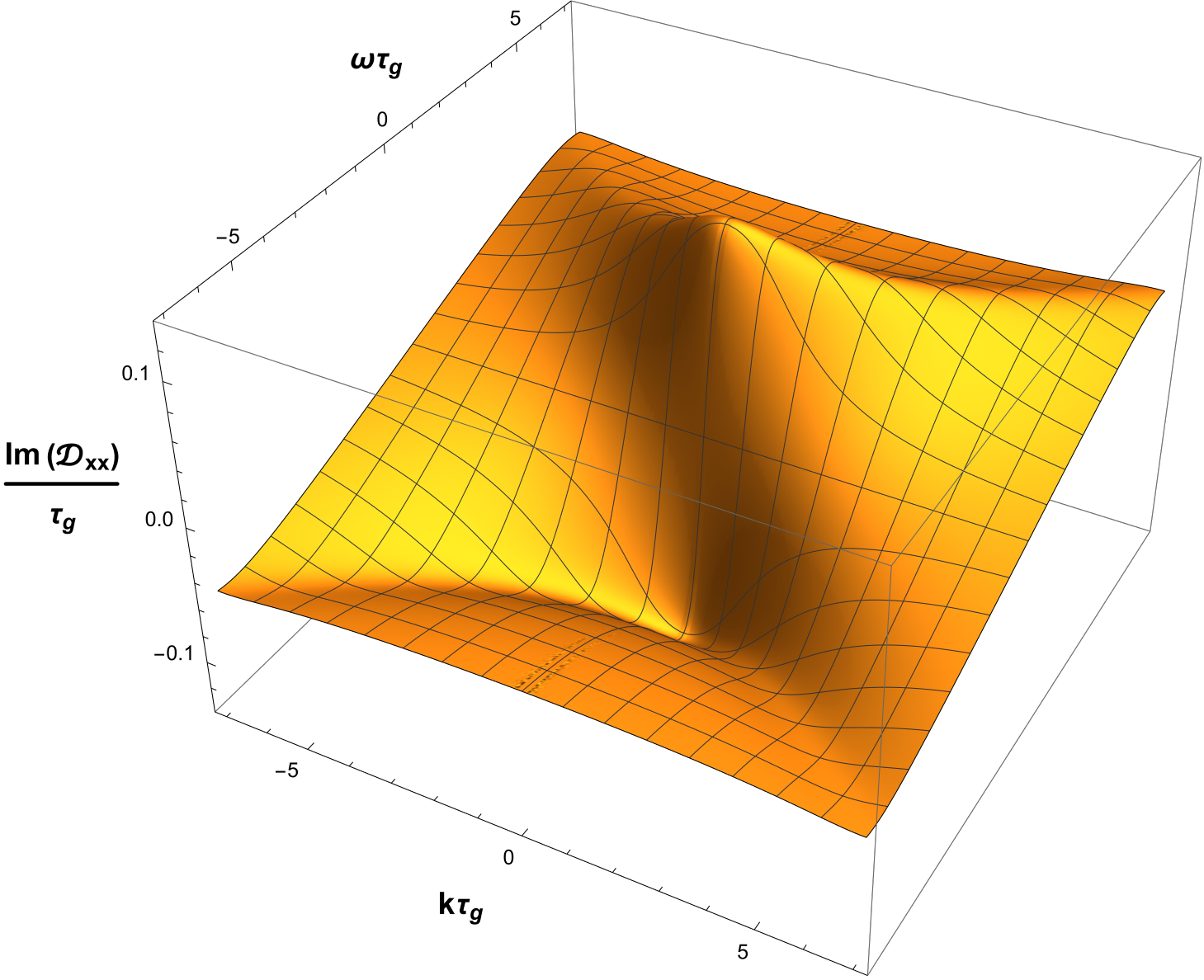}
\caption{Real (left) and imaginary (right) parts of the resummed diffusivities $\mathfrak{D}_{yy}$ (top) and $\mathfrak{D}_{xx}$ (bottom) in RTA, for real $\omega$ and $k$. As either frequency or wave number increases, the nonequilibrium flux response is suppressed, since particles average over increasingly rapid oscillations of the perturbation during a relaxation time. Increasing $\omega$ also produces a pronounced phase lag between flux and density, reflecting the finite response time. The dependence on $k$ is milder, both in magnitude and phase, because particles arriving from different directions experience and respond differently to the spatial gradient. The two diffusivities have very similar profiles, reflecting the same underlying physics of ballistic propagation over a relaxation time, but are not proportional: $\mathfrak{D}_{yy}/\mathfrak{D}_{xx}$ varies approximately between $0.7$ and $0.8$, and equals $3/4$ exactly along the line $k=0$.}
\label{fig:RTAFourier}
\end{figure}

\newpage
\section{Gradient expansion of Fokker-Planck kinetic theory}

\subsubsection{Outline of the theory}

We consider a non-degenerate ensemble of massless spinless particles in one spatial dimension, which do not interact with each other but scatter off a background medium, held at constant inverse temperature $\beta=1/T$. In the limit where the collisions with the medium are infinitely frequent but infinitely soft (Brownian motion in momentum space), the particle distribution function evolves according to the following linearized equation \cite{Debbasch,DunkelHanggi}:
\begin{equation}\label{BoltzmannFokkerPlanck}
(\partial_t +v\partial_x)\delta f =\dfrac{1}{\beta^2 \mathfrak{D}} \partial_p (\partial_p \delta f + \beta v \delta f) \, .
\end{equation}
Here, $\mathfrak{D}$ is the ordinary (non-resummed) particle diffusivity, and $v=\text{sign}(p)$ is the velocity of the particle with momentum $p$. The quadratic free-energy density reads
\begin{equation}
2\Delta\mathfrak{F}=\int_\mathbb{R} \dfrac{dp}{2\pi \beta} e^{\beta|p|} (\delta f)^2 \, .
\end{equation}
Hence, if we define the rescaled variable $\Psi{=}e^{\beta |p|/2}\delta f/\sqrt{2\pi \beta}$, the Onsager inner product reduces to $(\Phi,\Psi){=}\int_\mathbb{R} \Phi^* \Psi dp$, so that $\mathcal{H}=L^2(\mathbb{R})$. In these new variables, equation \eqref{BoltzmannFokkerPlanck} reduces to
\begin{equation}
(\partial_t +v\partial_x)\Psi=\dfrac{1}{\beta^2 \mathfrak{D}} \left[\partial_p^2 \Psi +\beta\delta(p) \Psi -\dfrac{\beta^2}{4}\Psi \right],
\end{equation}
from which we read off the relevant operators: $\E=v$ and $\s=\frac{1}{\beta^2 \mathfrak{D}} \left[-\partial_p^2 -\beta\delta(p) +\frac{\beta^2}{4} \right]$. The collision kernel $\s$ is a one-dimensional Schr\"{o}dinger operator \cite{GavassinoFokkerPLanck:2026zsz}, whose kernel is spanned by the unit vector
\begin{equation}
e_n =\dfrac{e^{-\beta|p|/2}}{||e^{-\beta|p|/2}||} \, .
\end{equation}
Hence, there is a single conservation law, with density $\rho_n=(e_n,\Psi)$ and corresponding flux $J_n=(e_n,v\Psi)$, whose local equilibrium part $J_n^\text{l.e.}=(e_n,ve_n)$ vanishes. The corresponding conserved charge is the particle number, since one finds that $\delta n=\int \frac{dp}{2\pi} \delta f=\frac{1}{\sqrt{\pi}}(e_n,\Psi)$ and $\delta J=\int \frac{dp}{2\pi} v\delta f=\frac{1}{\sqrt{\pi}}(e_n,v\Psi)$. Thus, the quantity
\begin{equation}
\mathfrak{D}_{nn}(i\omega,ik)= (\PN v e_n,[\s_N+ik\E_N-i\omega]^{-1}\PN v e_n) 
\end{equation}
is the resummed diffusivity entering the balance law of the particle number.

\subsubsection{Evaluation of the resummed diffusivity}

In the Fokker-Planck case, the inverse of the compressed operator $\s_N+ik\E_N-i\omega$ is rather difficult to evaluate explicitly. Hence, we will follow a different route.

Consider the vector $\Phi=[\s+ik\E-i\omega]^{-1}e_n$. Since $\mathcal{H}=\text{span}(e_n)\oplus e_n^\perp$, we can write $\Phi=Me_n +\Phi_\perp$, for some $M\in \mathbb{C}$ and $\Phi_\perp \in e_n^\perp$. Hence, taking the inner product with $e_n$, we obtain
\begin{equation}\label{Matrixonewow}
M=(e_n,\Phi)=(e_n,[\s+ik\E-i\omega]^{-1}e_n) \, .
\end{equation}
On the other hand, we have $[\s+ik\E-i\omega](Me_n+\Phi_\perp)=e_n$. Applying $\PN$ on both sides, we obtain  the identity $\Phi_\perp=-ikM[\s_N+ik\E_N-i\omega]^{-1}\PN v e_n$, which, when plugged into the unprojected equation, gives the following result:
\begin{equation}\label{Dnnnoncopmpessedformula}
\mathfrak{D}_{nn}=-\dfrac{i\omega+M^{-1}}{(ik)^2} \, . 
\end{equation}
Thus, we only need to evaluate the matrix element \eqref{Matrixonewow}, which involves \textit{uncompressed} operators.

Writing $[\s+ik\E-i\omega]\Phi=e_n$ explicitly, we have the following inhomogeneous Schr\"{o}dinger problem:
\begin{equation}\label{schrdoinghiam}
\left[-\partial^2_p-\beta\delta(p)+\beta^2\left(\dfrac{1}{4}+ ik\mathfrak{D}v- i\omega\mathfrak{D} \right)\right]\Phi =\mathfrak{D}\beta^2 \sqrt{\dfrac{\beta}{2}} e^{-\beta|p|/2} \, .  
\end{equation}
We must find solutions that decay at infinity.
Following the standard textbook procedure, we decompose \eqref{schrdoinghiam} as
\begin{equation}\label{schroddecomposedwell}
\begin{split}
&\left[-
\partial_p^2+\beta^2 r_-^2\right]\Phi=\beta^2 z \, e^{\beta p/2} \, , \qquad \qquad \text{for }p<0 \, ,\\
&\left[-\partial_p^2+\beta^2 r_+^2\right]\Phi=\beta^2 z \, e^{-\beta p/2}\, , \qquad \qquad \text{for }p>0 \, , \\
& \Phi(0^-)=\Phi(0^+) \, , \qquad\qquad \Phi'(0^+)-\Phi'(0^-)+\beta\Phi(0)=0 \, ,
\end{split}
\end{equation}
where we have introduced, for convenience, the quantities 
\begin{equation}
r_\pm=
\sqrt{\frac{1}{4}\pm ik\mathfrak{D} -i\omega \mathfrak{D}} \, , \qquad \qquad 
z=\mathfrak{D}\sqrt{\dfrac{\beta}{2}} \, .
\end{equation}
Plugging the ansatz solution\footnote{The signs of the exponents in \eqref{ansatz} were chosen to guarantee decay at infinity of $\Phi$ for $i\omega$ and $ik$ within the causal diamond.}\vspace{-0.2cm}
\begin{equation}\label{ansatz}
\Phi=\begin{cases}
A_- e^{\beta p/2} +B_- e^{r_- \beta p} & \text{for }p<0, \\
A_+ e^{-\beta p/2}+B_+ e^{-r_+ \beta p} & \text{for }p>0, \\
\end{cases}  
\end{equation}
into \eqref{schroddecomposedwell}, we obtain
\vspace{-0.3cm}
\begin{equation}
A_\pm = \dfrac{4 z}{4 r_\pm^2 -1} \, , \qquad B_\pm=\pm\frac{8 z \left(r_+^2-r_-^2\right)}{(r_++r_--1)\left(4 r_\pm^2-1\right) (2 r_\mp+1)} \, .
\end{equation}
Taking the inner product $M=(e_n,\Phi)$, we obtain
\vspace{-0.2cm}
\begin{equation}
M= \frac{A_-+A_++\frac{2 B_-}{2 r_-+1}+\frac{2 B_+}{2 r_++1}}{\sqrt{2\beta}} \, ,
\end{equation}
which plugged into \eqref{Dnnnoncopmpessedformula} finally gives us the resummed diffusivity:
\begin{equation}
\mathfrak{D}_{nn}(i\omega,ik)= \frac{ik \left[\sqrt{1{-}4 \mathfrak{D} (i\omega{+}ik)}-\sqrt{1{-}4 \mathfrak{D} (i\omega{-}ik)}\right]-i\omega \left[\sqrt{1{-}4 \mathfrak{D} (i\omega{+}ik)}+\sqrt{1{-}4 \mathfrak{D} (i\omega{-}ik)}-2\right]}{(i\omega)^2 \left[\sqrt{1{-}4 \mathfrak{D} (i\omega{+}ik)}+\sqrt{1{-}4 \mathfrak{D} ( i\omega{-}ik)}+2\right]-ik i\omega \left[\sqrt{1{-}4 \mathfrak{D} (i\omega{+}ik)}-\sqrt{1{-}4 \mathfrak{D} ( i\omega{-}ik)}\right]-4 (ik)^2} \, .
\end{equation}
Plugging this expression into the conservation law $i\omega+\mathfrak{D}_{nn}(ik)^2=0$, we recover Fick's law of diffusion $i\omega+\mathfrak{D}(ik)^2=0$, in agreement with \cite{GavassinoFokkerPLanck:2026zsz,GavassinoDiffusionCompatible:2026tvy}.
The profile of $\mathfrak{D}_{nn}$ on the Fourier plane (i.e. for real $\omega$ and $k$) is shown in figure \ref{fig:FokkerPlanckFourier}.

\begin{figure}[h!]
    \centering
\includegraphics[width=0.48\linewidth]{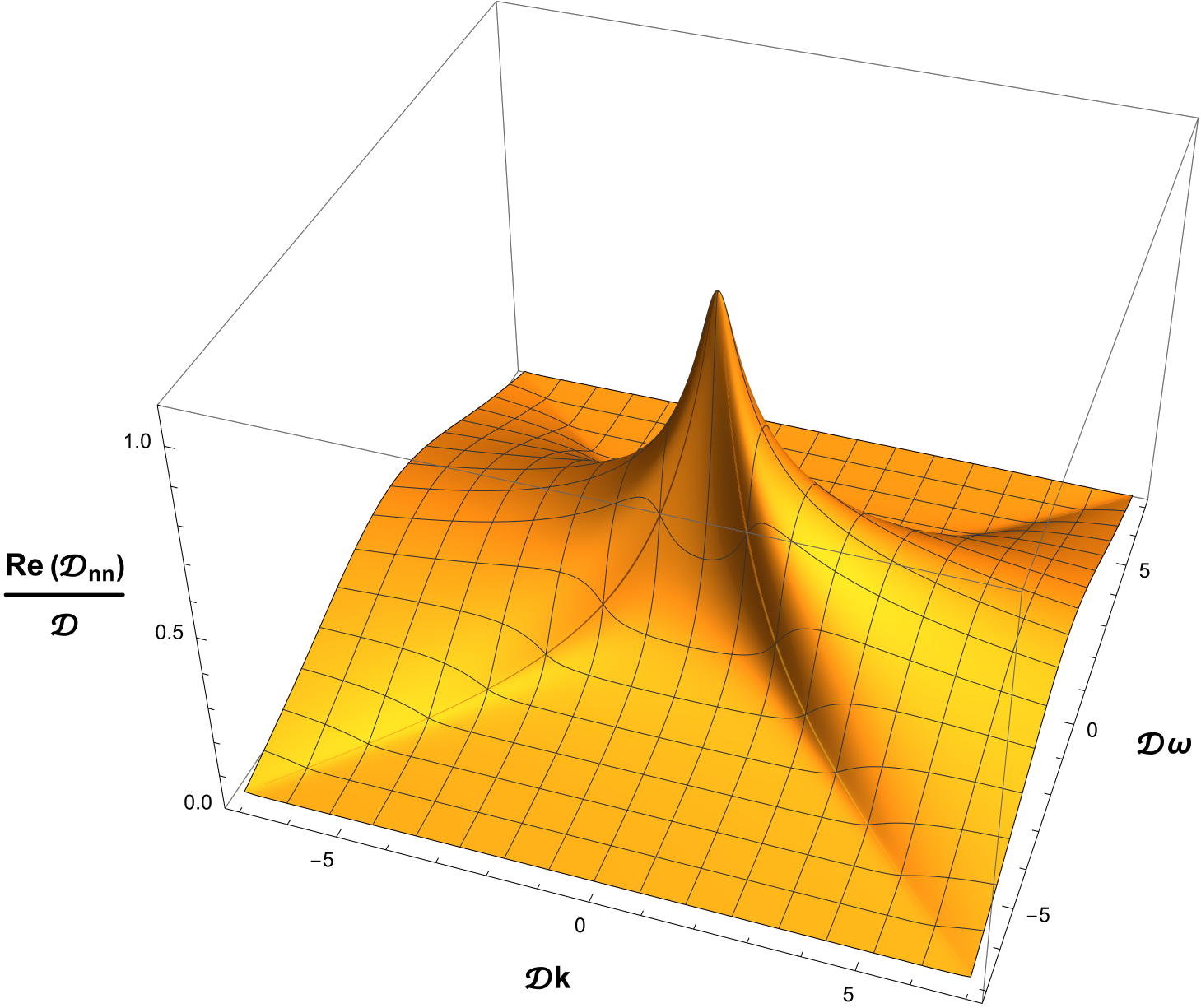}\hspace{0.02\linewidth}\includegraphics[width=0.48\linewidth]{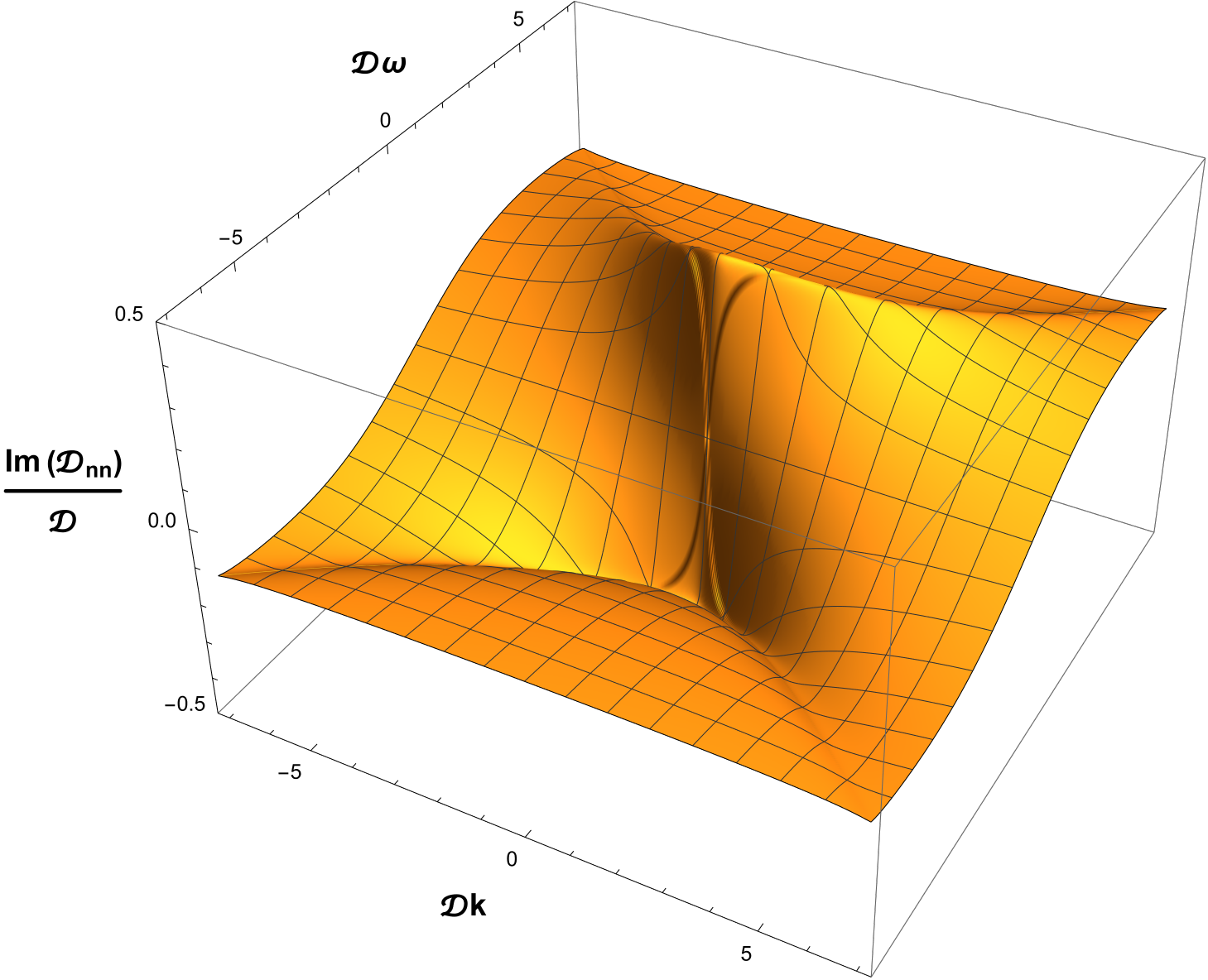}
\caption{Real (left) and imaginary (right) parts of the resummed diffusivity $\mathfrak{D}_{nn}$ in Fokker-Planck kinetic theory in one spatial dimension, for real $\omega$ and $k$. The trends are very similar to those found in RTA (see figure \ref{fig:RTAFourier}). The apparent loss of smoothness along the lines $\omega=\pm k$ is a numerical artifact, caused by evaluating an indeterminate form ``$0/0$''. Theorem \ref{theoFourier} guarantees that $\mathfrak{D}_{nn}$ is regular everywhere on the Fourier plane.}
\label{fig:FokkerPlanckFourier}
\end{figure}

\label{lastpage}
\end{document}